\documentclass{llncs}
\usepackage{amsmath,amsfonts,amssymb}
\usepackage{graphicx}
\usepackage{cite}
\usepackage{cyrillic}

\begin{document}

\sloppy 

\def\pal{{\cal P}}
\def\rich{{\cal R}}
\def\oddpal{{\cal O}}
\def\evenpal{{\cal E}}
\def\antipal{{\cal A}}
\def\E{{\sf E}}
\def\O{{\sf O}}
\linespread{0.93}

\title{On the Combinatorics of Palindromes and Antipalindromes}
\author{Chuan Guo\inst{1} \and Jeffrey Shallit\inst{1} \and Arseny M. Shur\inst{2}}

\institute{University of Waterloo, Ontario, Canada\\ \email{\{c3guo,shallit\}@uwaterloo.ca}
\and
Ural Federal University, Ekaterinburg, Russia\\ \email{arseny.shur@urfu.ru}}

\maketitle

\begin{abstract}
We prove a number of results on the structure and enumeration of palindromes and antipalindromes.  In particular, we study conjugates of palindromes, palindromic pairs, rich words, and the counterparts of these notions for antipalindromes.
\end{abstract}

\section{Introduction}

Combinatorial and algorithmic studies of palindromes can be traced back to the 1970's, when they were considered as a promising tool to construct a ``hard'' context-free language, which cannot be recognized by a linear-time random access machine. Nevertheless, palindrome-based languages were proved to be linear-time recognizable \cite{KMP77,GaSe78,KRS15}. Recent topics of interest in the study of palindromes include, for example, \emph{rich words} (containing the maximum possible number of distinct palindromes; see \cite{DJP01,GJWZ09,BLL09}) and \emph{palstars} (products of even-length palindromes; see \cite{KMP77,RSW11,RiSh14}). Also, there is a popular modification of the notion of palindrome, where the reversal of a word coincides not with the word itself, but with the image of the word under a certain involution of the alphabet; see, e.g., \cite{KaMa10,PeSt12}. In the binary case, there is a unique such modification, called an \emph{antipalindrome}.

In this paper, we aim to fill certain gaps in the knowledge on combinatorics of palindromes and antipalindromes. The four subsequent sections are mostly independent. In Section~\ref{sectwo}, we study the distribution of palindromes among conjugacy classes and enumerate conjugates of palindromes. Section~\ref{secthree} is devoted to the words which are products of two palindromes; we prove some characterizations of this class of words and show that the number of $k$-ary words that are products of two odd-length palindromes is exactly $k$ times the number of $k$-ary words of the same length that are products of
two even-length palindromes. In Section~\ref{secfour}, we analyze the growth function for the language of binary rich words. We give the first nontrivial lower bound, of order $C^{\sqrt n}$ for a constant $C$, and provide some empirical evidence that this growth function does indeed have subexponential growth. Finally, in Section~\ref{secfive} we focus on antipalindromes. We show that antipalindromes share many common properties with palindromes, with a notable exception: the notion of a rich word becomes trivial.

\paragraph{Definitions and notation.}

We study finite words over finite alphabets, using the array notation $w=w[1..n]$ when appropriate. The notions of prefixes, suffixes, factors, periods and (integer) powers are as usual. We write $|w|$ for the length of $w$ and $\varepsilon$ for the empty word. For two words $v$ and $w$ of length $n$, their \emph{perfect shuffle} $v\sha w$ is the word $v[1]w[1]v[2]w[2]\cdots v[n]w[n]$. Thus, for example, ${\tt clip} \, \sha \, {\tt aloe} = {\tt calliope}$. Given a word $w$, let $w^R$ denote its reversal (e.g., $({\tt stressed})^R = {\tt desserts}$).  A word $w$ is a {\it palindrome} if $w = w^R$. A word is \emph{primitive} if it is not an integer power of a shorter word. Two words $u$ and $v$ are \emph{conjugates} if $u=xy$ and $v=yx$ for some words $x$ and $y$. Conjugacy is an equivalence relation. The following lemma is folklore.

\begin{lemma} \label{conjnum}
Let $u=z^i$ for a primitive word $z$. Then the conjugacy class of $u$ contains exactly $|z|$ words.
\end{lemma}

We use two basic properties of periodic words due to Lyndon and Sch\"utzenberger \cite{LySch62}.

\begin{lemma} \label{lysch}
(i) For any nonempty words $u$ and $v$, the equality $uv=vu$ holds if and only if $u=z^i$ and $v=z^j$ for some word $z$ and positive integers $i,j$.\\
(ii) For any nonempty words $u$, $v$, and $w$, the equality $uw=wv$ holds if and only if $u=xy$, $v=yx$, $w=(xy)^ix$ for some words $x\ne\varepsilon$ and $y$, and nonnegative integer $i$.
\end{lemma}

For a language $L$ over an alphabet $\Sigma$, its \emph{growth function} (also called combinatorial complexity or census function) is defined to be $C_L(n)=|L\cap\Sigma^n|$.

Below, we list some basic properties of palindromes.

\begin{proposition} \label{pro:xn}
For all integers $m, n \geq 1$, the word $x^m$ is a palindrome if and only if $x^n$ is a palindrome.
\label{prop1} \end{proposition}

\begin{proposition} \label{pro:uvuv}
For all nonempty palindromes $u, v$, the word $uv$ is a palindrome if and only if both $u, v$ are powers of some palindrome $z$.
\label{prop2}
\end{proposition}

\begin{proof} 
Suppose $uv$ is a palindrome.  Then $uv = (uv)^R = v^R u^R = vu$. By Lemma~\ref{lysch} (i), $u = z^i$, $v = z^j$ for integers $i, j \geq 1$. By Proposition~\ref{prop1}, we have that $z$ is a palindrome. Conversely, if $u = z^i$ and $v = z^j$, we have $uv = z^{i+j}$, which is a palindrome by Proposition~\ref{prop1}.\qed
\end{proof}

\begin{proposition} \label{pro:sha}
The word $x$ is an even-length palindrome iff there exists a word $y$ such that $x = y \, \sha \, y^R$.
\end{proposition}




\section{Conjugates of Palindromes}
\label{sectwo}

Here we study the distribution of palindromes in conjugacy classes and count conjugates of palindromes.

\begin{theorem} \label{the:conj}
A conjugacy class contains at most two palindromes. A conjugacy class has two palindromes if and only if it contains a word of the form $(xx^R)^i$, where $xx^R$ is a primitive word and $i\ge 1$.
\end{theorem}

\begin{lemma} \label{lem:xxR}
Suppose $u\ne u^R$ and $uu^R=z^i$ for a primitive word $z$. Then $i$ is odd and $z=xx^R$ for some $x$.
\end{lemma}

\begin{proof}
If $i$ is even, then $uu^R=(z^{i/2})^2$. Hence, $u=u^R$, contradicting the conditions of the lemma. So $i$ is odd and then $|z|$ is even. Let $z=xx'$, where $|x|=|x'|$. We see that $x$ is a prefix of $u$ and $x'$ is a suffix of $u^R$. Hence, $x'=x^R$, as required.\qed
\end{proof}

\begin{proof}[of Theorem~\ref{the:conj}]
Let us prove that for any conjugacy class with two distinct palindromes, say $uv$ and $vu$, there exists a word $x$ and a number $i$ such that $xx^R$ is primitive, $uv=(xx^R)^i$, and $vu=(x^Rx)^i$. We use induction on $n=|uv|$. The base case is trivial, because such conjugacy classes do not exist for, say, $n=1$.

For the inductive step, assume $|u| \leq |v|$ without loss of generality. If $|u|=|v|$, then $v=u^R$. By Lemma~\ref{lem:xxR} we get $uv=(xx^R)^i$, $vu=(x^Rx)^i$ for a primitive word $xx^R$ and $i\ge 1$.

Now let $|u|<|v|$. Then $v$ begins and ends with $u^R$. Applying Lemma~\ref{lysch} (ii), we obtain $u^R=(st)^is$, where $s\ne\varepsilon$ and $i\ge 0$ and, respectively, $v=(st)^{i+1}s$. Looking at the central factor of the palindromes 
\begin{align}
uv &= (s^Rt^R)^is^R\cdot st\cdot s(ts)^i, \\
vu &= (st)^is\cdot ts\cdot s^R(t^Rs^R)^i,
\end{align}
we see that $st$ and $ts$ are also palindromes. If $t=\varepsilon$, then $s$ is a palindrome, implying $uv=vu$, which is impossible. If $st=ts$, then by Lemma~\ref{lysch} (i) both $s$ and $t$ are powers of some $z$. By Proposition~\ref{pro:xn}, $s$, $t$, and $z$ are palindromes, and then again $uv=vu$. Therefore, we obtain $st\ne ts$. So we can apply the inductive hypothesis to these two palindromes, getting $st=(xx^R)^j$, $ts=(x^Rx)^j$ for a primitive word $xx^R$. Then we can write
\begin{equation} \label{one} 
v=(xx^R)^{j(i+1)}s=s(x^Rx)^{j(i+1)}.
\end{equation}
Since all conjugates of the word $xx^R$ are distinct words by Lemma~\ref{conjnum}, $x$'s and their reversals occur in the same positions in both representations \eqref{one}. Then $s=(xx^R)^kx$ for some $k$, $0\le k<j$. Then we can easily compute $t$, $s^R$, and $t^R$ to get $uv=(xx^R)^{2j(i+1)}$, $vu=(x^Rx)^{2j(i+1)}$. Thus, we have finished the inductive step.

Note that a conjugacy class of a palindrome $(xx^R)^i$, where $xx^R$ is primitive, clearly contains a different palindrome $(x^Rx)^i$. To finish the proof of the theorem it remains to show that such a class contains no other palindromes. Indeed, consider the class of $w=(xx^R)^i$. It consists of $i$th powers of conjugates of $xx^R$ (see Lemma~\ref{conjnum}). Let $u$ and $v$ be such that $uv=xx^R$ and $(vu)^i$ is a palindrome. Then $vu$ is a palindrome by Proposition~\ref{pro:xn}. If $|u|\ne |v|$, say, $|u|<|v|$, then we apply the above argument to $uv$ and $vu$, getting $uv=(yy^R)^{2k}$ for some $y$ and $k$. But this is impossible, because $uv=xx^R$ is primitive. Hence, $|u|=|v|$, and thus $vu=x^Rx$. Thus, the class of $w$ contains exactly two palindromes.\qed
\end{proof}

\begin{corollary} \label{cor:012}
A conjugacy class of a word $w=z^m$, where $z$ is primitive and $m\ge1$, contains\\
(i) 0 or 1 palindrome, if $|z|$ is odd;\\
(ii) 0 or 2 palindromes, if $|z|$ is even. 
\end{corollary}

Let $n=p_1^{i_1}\cdots p_k^{i_k}$, where $p_1,\ldots,p_k$ are primes. Recall that the M\"obius function $\mu(n)$ equals $(-1)^k$ if $i_1=\cdots =i_k=1$ (i.e., if $n$ is \emph{square-free}) and $0$ otherwise.  The following result apparently first appeared in \cite{Nes99}.

\begin{lemma}
The number $\rho(k,n)$ of $k$-ary words of length $n$ that are both (i) primitive and (ii) a palindrome satisfies the following formula
\begin{equation} \label{primpal}
\rho(k,n)=\sum_{d|n} \mu(d) k^{\lfloor ((n/d)+1)/2 \rfloor}.
\end{equation}
\end{lemma}

We use \eqref{primpal} and Corollary~\ref{cor:012} to count the conjugates of palindromes. 
\begin{theorem}
The number of conjugates of $k$-ary palindromes of length $n$ is 
\begin{equation} \label{conjpal}
c(k,n) = \sum_{d|n}f(d)\cdot \rho(k,d), \text{ where }f(d) = 
\begin{cases}
d,& \text{ if $d$ is odd};\\
d/2,& \text{if $d$ is even}.
\end{cases}
\end{equation}
\end{theorem}

\begin{proof}
If $n$ is odd, it suffices to count the number of palindromes and multiply it by the number of distinct conjugates given in Lemma~\ref{conjnum}. Instead of non-primitive palindromes, we count their primitive roots, which are palindromes by  Proposition~\ref{pro:xn}. Thus we have $$ c(k,n) = \sum_{d|n} d\cdot \rho(k,d), $$ which is equivalent to \eqref{conjpal} because $d$ takes only odd values. If $n$ is even, some classes contain two palindromes. Let $w=z^i$, where $z$ is primitive. If $|z|$ is odd, then the class of $w$ contains 0 or 1 palindrome, while if $|z|$ is even, the class of $w$ contains 0 or 2 palindromes. In the latter case, we must divide the result of counting the conjugates of palindromes by $2$. This gives precisely \eqref{conjpal}.\qed 
\end{proof}

\section{Palindromic Pairs}
\label{secthree}

In this section we consider palindromic pairs, which are words factorizable into two palindromes. First we give a few easy characterizations of palindromic pairs, and then prove Theorem~\ref{enum} on the number of ``even'' and ``odd'' palindromic pairs.

Let $\pal$ be the set of palindromes over the alphabet $\Sigma$. Palindromic pairs are exactly the elements of $\pal^2$. Recall that $L_1/L_2 = \lbrace x \in \Sigma^* \mid \exists y \in L_2: xy\in L_1 \rbrace$.  

\begin{proposition}
$\pal/\pal = \pal^2$.
\label{thm1}
\end{proposition}

\begin{proof}
$\pal/\pal \subseteq \pal^2$:  Suppose $x \in \pal/\pal$.  Then there  exists a palindrome $y$ such that $xy$ is a palindrome.  If either $x$ or $y$ is empty, the result is clearly true.  Otherwise $xy = (xy)^R = y^R x^R = y x^R$. Then by Lemma~\ref{lysch} (ii) there exist $u \in \Sigma^+$, $v \in \Sigma^*$ and an integer $e \geq 0$ such that $x = uv$, $x^R = vu$, and $y = (uv)^i u$. From $x^R = vu$ we get $x = u^R v^R$.  But $x = uv$, so $u = u^R$ and $v = v^R$. So $u,v\in\pal$ and
then $x \in \pal^2$.

$\pal^2 \subseteq \pal/\pal$: If $x = uv$ with $u, v$ palindromes, then $uvu$ is a palindrome. Thus, taking $y = u$, we have $xy \in \pal$ and $y\in \pal$. Hence $x \in \pal/\pal$.
\qed \end{proof}

Call a word $x$ {\it credible} if it is a conjugate of its reverse $x^R$, like the English word {\tt referee}.

\begin{proposition}
The word $x$ is in $\pal^2$ if and only if it is credible.
\end{proposition}

\begin{proof}
Suppose $x \in \pal^2$, that is, that $x = uv$ where $u, v$ are palindromes. Then $x^R = v^R u^R = vu$, so $x$ is a conjugate of $x^R$.

Otherwise, assume $x$ is a conjugate of $x^R$.  Then there exist $u, v$ such that $x = uv$ and $x^R = vu$.    But $x^R = v^R u^R$, so $v = v^R$ and $u = u^R$, and $x$ is the product of two palindromes.\qed
\end{proof}

\begin{remark}
The growth function for $\pal^2$ was studied by Kemp \cite{Kemp82}, who
computed a precise formula and described its asymptotics. For the
binary alphabet, see sequence A007055 in the {\it On-Line Encyclopedia
of Integer Sequences} \cite{Slo}.
\end{remark}

A factorization of a palindromic pair $w$ is a pair of palindromes $u,v$ such that $uv=w$ and $v\ne\varepsilon$. The number of factorizations of a palindromic pair is described in the following theorem, also due to Kemp \cite{Kemp82}.

\begin{theorem} \label{the:pp1}
A palindromic pair $w$ has $m$ factorizations if and only if $w=z^m$ for a primitive word $z$.
\end{theorem}

A palindromic pair $w$ is even (resp., odd) if it can be factorized into two palindromes of even (resp., odd) length. Thus, an even-length palindromic pair is either even, or odd, or both, like the word 
\begin{equation} \label{aabaab}
aabaab=aa\cdot baab=aabaa\cdot b\,.
\end{equation} 
Note that in view of Theorem~\ref{the:pp1}, the last option applies to non-primitive words only. Let $\E(n,k)$ (resp., $\O(n,k)$) denote the number of even (resp., odd) $k$-ary palindromic pairs of length $n$.

\begin{theorem} \label{enum}
For all $n$ and $k$ we have $\O(n,k)=k\cdot\E(n,k)$.
\end{theorem}

The proof is based on several lemmas. The following lemma is a well-known corollary of the Fine-Wilf property \cite{FiWi65}:

\begin{lemma} \label{periods}
Let $p$ be the minimal period of a word $w$, $q$ be a period of $w$, and $p,q\le |w|/2$. Then $q$ is a multiple of $p$. 
\end{lemma}

We call an even-length word $w$ \emph{even-primitive} if it is not an integer power of a shorter even-length word. The difference between primitive and even-primitive words is clarified in the following lemma.

\begin{lemma} \label{evenprim}
An even-length non-primitive word is even-primitive iff it is the square of an odd-length primitive word.
\end{lemma}

\begin{proof} Let $w$ be an even-length non-primitive word. Then $w=z^m$, for a primitive word $z$ and an integer $m>1$. Such a pair $(z,m)$ is unique. Indeed, $|z|$ is a period of $w$, $|z|\le |w|/2$, and $|z|$ is not a multiple of a shorter period of $w$ due to primitivity. Then $|z|$ is the minimal period of $w$ by Lemma~\ref{periods}.

\emph{Necessity}. Since $w$ is even-primitive, $|z|$ is odd. Then $m$ is even, because $|w|=m|z|$. If $m\ge4$, then $w=(z^2)^{m/2}$ is not even-primitive. Hence, $m=2$, as desired.

\emph{Sufficiency} stems from the fact that $|z|$ is the minimal period of $w$.\qed
\end{proof}

A palindromic pair can have several factorizations according to Theorem~\ref{the:pp1}. However, the even-primitive words have the following useful property (cf.~\eqref{aabaab}).

\begin{lemma} \label{evenrepr}
An even-length even-primitive word has at most one factorization into two even-length palindromes and at most one factorization into two odd-length palindromes.
\end{lemma}

\begin{proof}
Let $u_1,u_2,v_1,v_2$ be even-length palindromes such that $u_1v_1=u_2v_2=w$ and $|v_1|<|v_2|$ (recall that $v_1\ne\varepsilon$ by the definition of factorization). Then $v_1$ is a suffix of $v_2$ and, as both these words are palindromes, a prefix of $v_2$ as well. By Lemma~\ref{lysch} (ii), $v_1=(xy)^sx$, $v_2=(xy)^{s+1}x$ for some words $x\ne\lambda$ and $y$, and some integer $s\ge0$; in addition, it is clear that both $x$ and $y$ are palindromes. Hence, $u_1=u_2xy$, and, since $u_1,u_2,x,y$ are palindromes, $u_1=yxu_2$. Applying Lemma~\ref{lysch} (ii) again, we finally get $w=(yx)^t$ for some $t\ge2$. Since $|yx|=|v_2|-|v_1|$ is even, $w$ is not even-primitive.

The same argument works for odd-length palindromes $u_1,u_2,v_1,v_2$.\qed
\end{proof}

\begin{lemma} \label{evenpair}
Let $z$ be an even-primitive word, $m>1$ be an integer. Then $z^m$ is an even (resp., odd) palindromic pair iff $z$ is.
\end{lemma}

\begin{proof}
\emph{Necessity}. If $z^m=uv$ for palindromes $u$ and $v$, then $u=(xy)^tx$, $v=y(xy)^{m-t-1}$ for some words $x,y$ and some integer $t$ such that $xy=z$ and $0\le t\le m{-}1$. Since $u$ and $v$ are palindromes, $x$ and $y$ are also palindromes, implying that $z$ is a palindromic pair. Since $|z|=|xy|$ is even by definition, the numbers $|x|$ and $|u|$ (resp., $|y|$ and $|v|$) have the same parity, whence the result. 

\emph{Sufficiency}. If $z=xy$ for palindromes $x$ and $y$, then $z^m=x\cdot y(xy)^{m-1}$. The word $y(xy)^{m-1}$ is a palindrome and its length has the same parity as $y$.\qed
\end{proof}

In what follows we suppose that the size $k>1$ of the alphabet is fixed (for the unary alphabet the theorem is straightforward), and write $\E(n)$ (resp., $\O(n)$) for $\E(n,k)$ (resp., $\O(n,k)$). We also assume that $n$ is even, because $\E(n)=\O(n)=0$ for odd $n$. Let $\E'(n)$ (resp., $\O'(n)$) be the number of even-primitive even (resp., odd) palindromic pairs of length $n$.

\begin{lemma} \label{main}
If $\O'(n)=k\cdot\E'(n)$ for all $n$, then $\O(n)=k\cdot\E(n)$ for all $n$.
\end{lemma}

\begin{proof}
Any word $w$ of length $n$ can be uniquely represented as $w=z^{n/2d}$, where $2d$ is a divisor of $n$ and the word $z$ of length $2d$ is even-primitive. Then by Lemma~\ref{evenpair} the number of even (resp., odd) palindromic pairs of this form is equal to the number of even (resp., odd) palindromic pairs $z$. The latter number is exactly $\E'(2d)$ (resp., $\O'(2d)$). Hence,
$$
\O(n)=\sum_{2d\mid n}\O'(2d)=k\cdot\sum_{2d\mid n}\E'(2d)=k\cdot \E(n).\quad \qed
$$
\end{proof}

\begin{proof}[of Theorem~\ref{enum}]
In view of Lemma~\ref{main}, it suffices to prove the equality $\O'(n)=k\cdot\E'(n)$ for all even $n$.

Let $P_i(n)$ be the set of all palindromic pairs $w=uv$ such that $|w|=n$, $|u|=i$. Then the sets $P_e(n)$ of all even palindromic pairs and $P_o(n)$ of all odd palindromic pairs of length $n$ can be written as
\begin{subequations}
\begin{align}
P_e(n)&=P_0(n)\cup P_2(n)\cup\cdots\cup P_{n-2}(n)\label{unioneven}\\
P_o(n)&=P_1(n)\cup P_3(n)\cup\cdots\cup P_{n-1}(n)\label{unionodd}
\end{align}
\end{subequations}
In general, the sets $P_i(n)$ may intersect, but Lemma~\ref{evenrepr} tells us that an even-primitive word belongs to at most one set of \eqref{unioneven} and at most one set \eqref{unionodd}. Furthermore, let $2d$ be a divisor of $n$, and consider an even palindromic pair of the form $w=z^{n/2d}$, where $z$ is even-primitive. Then $z$ is an even palindromic pair by Lemma~\ref{evenpair}. By Lemma~\ref{evenrepr}, $z$ has a unique factorization into two even-length palindromes, say, $z=xy$. Hence, $w$ has exactly $n/2d$ factorizations into two even-length palindromes (see  the proof of Lemma~\ref{evenpair}):
$$
w=x\cdot y(xy)^{n/2d-1}=xyx\cdot y(xy)^{n/2d-2}=(xy)^{n/2d-1}x\cdot y.
$$
Thus, $w$ belongs to exactly $n/2d$ sets \eqref{unioneven}. In the same way, an odd palindromic pair of the form $w=z^{n/2d}$ belongs to exactly $n/2d$ sets \eqref{unionodd}. To get the value of $\E'(n)$ (resp., $\O'(n)$), we must sum up the cardinalities of all sets \eqref{unioneven} (resp., \eqref{unionodd}) and subtract the total contribution of the words that are not even-primitive. Since the number of even (resp., odd) palindromic pairs of the form $w=z^{n/2d}$ is the same as the number of even-primitive even (resp., odd) palindromic pairs $z$ (of length $2d$), we have
\begin{equation}\label{EO1}
\E'(n)=\!\!\sum_{i=0}^{n/2-1}\!\!\#P_{2i}(n)-\sum_{\substack{2d\mid n,\\2d<n}}\frac{n}{2d}\cdot\E'(2d), \
\O'(n)=\!\!\sum_{i=0}^{n/2-1}\!\!\#P_{2i+1}(n)-\sum_{\substack{2d\mid n,\\2d<n}}\frac{n}{2d}\cdot\O'(2d).
\end{equation}
Now let us compute $\#P_i(n)$. If $i=2j$, then to determine a word from $P_i(n)$ we can arbitrarily choose the first $j$ letters to fix the first palindrome and the last $(n/2-j)$ letters to fix the second one. In total, we have $n/2$ letters, concluding that $\#P_i(n)=k^{n/2}$, independent of $i$. On the other hand, if $i=2j{+}1$, then to determine a word from $P_i(n)$ we must choose the first $j{+}1$ letters to fix the first palindrome and the last $(n/2-j)$ letters to fix the second one (the central letters of both odd-length palindromes must be chosen). Thus, in this case we have $\#P_i(n)=k^{n/2{+}1}$, also independent of $i$. Therefore, we rewrite \eqref{EO1} as
\begin{equation}\label{EO2}
\E'(n)=\frac{n}{2}\cdot k^{n/2}-\sum_{\substack{2d\mid n,\\2d<n}}\frac{n}{2d}\cdot\E'(2d), \quad
\O'(n)=\frac{n}{2}\cdot k^{n/2+1}-\sum_{\substack{2d\mid n,\\2d<n}}\frac{n}{2d}\cdot\O'(2d).
\end{equation}
Finally, we prove the equality $\O'(n)=k\cdot\E'(n)$ by induction. The base case is $n=2$. All words of length 2 are even-primitive, and all of them are odd palindromic pairs. In contrast, the only even palindromic pairs of length 2 are palindromes (of the form $aa$). Thus, $\O'(2)=k^2$, $\E'(2)=k$.
For the inductive step, note that the ratio of the first terms for $\O'(n)$ and $\E'(n)$ in \eqref{EO2} is $k$, and the inductive assumption implies that the ratio of the second terms (i.e., sums) is also $k$. The theorem is proved.\qed
\end{proof}

\section{Rich Words}
\label{secfour}

Recall that a word of length $n$ is (palindromic) rich if it contains $n$ distinct nonempty palindromic factors. Basic properties of rich words were proved in \cite{DJP01,GJWZ09}. Some of them are collected in the following

\begin{proposition} \label{pro:rich1}
\begin{itemize}
\item[(i)] A factor of a rich word is rich.
\item[(ii)]  A reversal of a rich word is rich.
\item[(iii)]  A word $w$ is rich if and only if the longest palindromic suffix of any its prefix $w[1..i]$ has no other occurrences in $w[1..i]$.
\end{itemize}
\end{proposition}

Let $\rich$ be the language of binary rich words. Here we analyze its growth function (sequence A216264 in the OEIS \cite{Slo}), starting with a lower bound.  This is, to our knowledge, the first nontrivial (not polynomial) lower bound obtained for the number of binary rich words.

Every word has a unique \emph{run-length encoding} of the form $a_1^{s_1}a_2^{s_2}\cdots a_k^{s_k}$, for $s_1, \ldots, s_k \ge 1$, where the letters $a_i$ and $a_{i+1}$ are distinct for any $i$.  We call each term $a_i^{s_i}$ a \emph{block}.  Consider the language $I$ of binary words whose run-length encoding satisfies $s_i\le s_{i+2}$ for all $i=1,\ldots,k{-}2$. The language $I$ is close to the languages of intermediate growth studied
in \cite{Sh09dam}.

\begin{theorem} \label{the:rich1}
$I\subseteq\rich$.
\end{theorem}

\begin{proof}
Let $w=0^{s_1}1^{s_2}0^{s_3}\cdots a_k^{s_k}\in I$ for $s_1, \ldots, s_k \geq 1$.  We prove the richness of all the prefixes of $w$ (including $w$ itself) by induction on the length of $w$.

The base case is trivial. Now assume that $w[1..i]$ is rich and $w[i]=0$ (the case $w[i]=1$ is analysed in the same way). We add the letter $w[i{+}1]$ and search for the ``new'' palindromic suffix described in Proposition~\ref{pro:rich1} (iii). If $w[1..i]=0^i$, there is nothing to prove: the new suffix is $0^{i+1}$ if $w[i{+}1]=0$ and $1$ if $w[i{+}1]=1$. So let $w[1..i]=0^{s_1}1^{s_2}\cdots 0^{s'_l}$, where $s'_l\le s_l$ and $1<l\le k$.

If $s'_l\ge s_{l-2}$ and $w[i{+}1]=0$, then the new palindromic suffix is $0^{s'_l+1}$; if $s'_l> s_{l-2}$ and $w[i{+}1]=1$, then the new suffix is $10^{s'_l}1$. Thus, two nontrivial cases remain.

\emph{Case 1}: $s'_l< s_{l-2}$; clearly then $w[i{+}1]=0$. The word $w[1..i]$ ends with a palindrome $0^{s'_l}$, but it has more occurrences in $w[1..i]$. Hence, it is not the longest palindromic suffix. Then this suffix, which we denote by $v$, begins with $0^{s'_l}$ and ends with $1^{s_{l-1}}0^{s'_l}$. If
$v=0^{s'_l}1^{s_{l-1}}0^{s'_l}$, then $0v$ is a suffix of $w[1..i]$ because $s'_l< s_{l-2}$. Hence, the palindrome $0v0$ is a suffix of $w[1..i{+}1]$, and it has no earlier occurrences because $v$ occurs in $w[1..i]$ only once.

If $v$ intersects more than three blocks, then $v=0^{s'_l}(1^{s_{l-1}}0^{s_{l-2}})^k1^{s_{l-1}}0^{s'_l}$ for some $k\ge 1$. If $v$ is preceded by 0 in $w$, then the suffix $0v0$ of $w[1..i{+}1]$ is a new palindrome, as before. Finally, if $v$ is preceded by $1$, consider the palindrome $v'=0^{s'_l}(1^{s_{l-1}}0^{s_{l-2}})^{k-1}1^{s_{l-1}}0^{s'_l}$. The word $w[1..i{+}1]$ ends with $0v'0$, and this is a new palindrome, because $0v'$ occurs in $w[1..i]$ only as a suffix. (Indeed, $0v'$ begins with $0^{s'_l+1}$, and there is no large enough block of zeroes to the left of $v$.)

\emph{Case 2}: $s'_l= s_{l-2}$ and $w[i{+}1]=1$. Since $0^{s'_l}1^{s_{l-1}}0^{s'_l}$ is a palindromic suffix of $w[1..i]$, the longest such suffix is $v=(0^{s'_l}1^{s_{l-1}})^k0^{s'_l}$ for some $k\ge 1$.  Clearly, $1v$ is a suffix of $w[1..n]$, implying that $w[1..i{+}1]$ ends with the new palindrome $1v1$.

Thus, in all cases, $w[1..i{+}1]$ ends with a new palindrome. The inductive step is finished.\qed
\end{proof}

\begin{theorem} \label{the:rich2}
The growth function of the language of binary rich words satisfies
\begin{equation}\label{eq:CRn}
\ln(C_{\rich}(n))\ge \frac{2\pi}{\sqrt3}\cdot\sqrt{n}-O(\ln n).
\end{equation}
\end{theorem}

\begin{proof}
Let $p(n)$ [$p(n,k)$] denote the number of integer partitions (resp., partitions with exactly $k$ parts) of $n$. There is a natural injection of partitions of $n$ into words of length $n$: a partition $s_1+\cdots+s_k=n$, where the parts are written in increasing order, defines the word $w=0^{s_1}1^{s_2}0^{s_3}\cdots a_k^{s_k}$. Note that $w\in I$, implying $C_I(n)\ge p(n)$. By the famous Hardy-Ramanujan-Uspensky formula, 
\begin{equation}\label{eq:HRU}
p(n)\sim \frac{e^{\pi\sqrt{2n/3}}}{4n\sqrt{3}} \text{ as }n\to\infty.
\end{equation}
So we already have the bound similar to \eqref{eq:CRn}, but with a smaller constant in the leading term. To get the desired constant, we assume $n$ to be even (since $C_I(n)$ is an increasing function, substituting the bound obtained for $n=2m$ for $n=2m{+}1$ gives the general bound of the same order of growth). Note that any pair of partitions $s_1+\cdots+s_k=t_1+\cdots +t_k=n/2$ can be mapped to the word $0^{s_1}1^{t_1}\cdots 0^{s_k}1^{t_k}\in I$, and this map is injective. Then we have
\begin{equation}\label{eq:CIn}
C_I(n)>\sum_{k=1}^{n/2-1} \big(p(n/2,k)\big)^2>\big(\max_k p(n/2,k)\big)^2> \frac {4\big(p(n/2)^2\big)}{n^2}.
\end{equation}
Substituting \eqref{eq:HRU} and taking logarithms, we obtain \eqref{eq:CRn}.\qed
\end{proof}

\begin{remark}
More precise estimates of $C_I(n)$ cannot give better bounds for $C_\rich(n)$, because all inequalities in \eqref{eq:CIn} affect only the $O$-term in \eqref{eq:CRn}. Using Proposition~\ref{pro:rich1}\,(i,ii), one can extend $I$, closing it under taking factors and reversals; but the effect of such extensions is swallowed by the $O$-term as well. 
\end{remark}

How good is this lower bound, which is roughly $37^{\sqrt{n}}$ divided by a polynomial? It is unclear, because no good upper bound for $C_\rich(n)$ has yet been obtained. The function $C_\rich(n)$ is submultiplicative due to Proposition~\ref{pro:rich1} (i), so its growth rate $\lim_{n\to\infty}(C_\rich(n))^{1/n}$ is majorized by any value $(C_\rich(n))^{1/n}$, according to Fekete's lemma \cite{Fek23}. Also, the ratio $\frac{C_\rich(n)}{C_\rich(n-1)}$ gives us a clue about the growth of this function. Until recently, the number of known values of $C_\rich(n)$ was quite small ($n\le 25$ on OEIS, posted by the second author). But short rich words constitute a substantial share of all short words, giving us $(C_\rich(25))^{1/25}\approx 1.818$ and $\frac{C_\rich(25)}{C_\rich(24)}\approx 1.599$. Recently Mikhail Rubinchik (personal communication) invented a new technique and computed this sequence up to $n=60$.  From his calculations we get
$(C_\rich(60))^{1/60}\approx 1.605$ and $\frac{C_\rich(60)}{C_\rich(59)}\approx 1.394$. Such a fast drop is unusual for exponential growth functions of languages closed under factors (see \cite{Sh12rev} and the references therein). So this is an argument in favor of a subexponential growth of $C_\rich(n)$.

A much stronger argument is provided in Table~\ref{nsqrtn} below: $C_\rich(n)\le n^{\sqrt{n}}$ for $4\le n\le 60$, and, moreover, the function $ n^{\sqrt{n}}$ seems to grow faster. So we propose the following conjecture, which implies that the bound of Theorem~\ref{the:rich2} is quite reasonable.

\begin{conjecture}
One has $C_\rich(n)=O\big(\frac{n}{g(n)}\big)^{\sqrt{n}}$, or, equivalently, $\ln (C_\rich(n))=O(\sqrt{n}(\ln n - g(n))$, for some infinitely growing function $g(n)$.
\end{conjecture}

\begin{table}[!thb]
\vspace*{-4mm}
\caption{Number of rich words compared to the $n^{\sqrt{n}}$ function.}\label{nsqrtn}
\centerline{
\tabcolsep=7pt
\begin{tabular}{|c|ccc|}
\hline
$n$ & $C_\rich(n)$ &  $n^{\sqrt{n}}$ & Ratio\\
\hline
4 &  $16$ & $16$ & $1$ \\
5 &  $32$ & $\approx36.55$ & $0.875$\\
\hline
$\cdots$ & & &  \\
\hline
25 & $3\;089\;518$ & $\approx9.766\cdot 10^6$& $0.335$ \\
26 & $4\;903\;164$ & $\approx1.641\cdot 10^7$& $0.316$ \\
\hline
$\cdots$ & & & \\
\hline
59 &$1\;530\;103\;385\;844$  & $\approx4.001\cdot 10^{13}$ & $0.038$ \\
\hline
60 &$2\;132\;734\;033\;216$  & $\approx5.936\cdot 10^{13}$ & $0.036$ \\
\hline
\end{tabular}
}
\vspace*{-2mm}
\end{table}

\section{Antipalindromes}
\label{secfive}

In this section, the alphabet is $\{0,1\}$. For a word $x \in \{ 0, 1
\}^*$, its \emph{negation} $\overline{x}$ is obtained by changing each
$0$ in $x$ to $1$ and vice versa.  A word $x$ is an {\it
antipalindrome} if $x = \overline x^R$.  Thus, for example, {\tt
001011} is an antipalindrome; note that all antipalindromes have even
length. Let $\antipal$ denote the set of all antipalindromes. There are
definite similarities between the properties of palindromes and
antipalindromes. The following analogs of
Propositions~\ref{pro:xn}-\ref{pro:sha} are straightforward.

\begin{proposition} \label{prop8}
For all integers $m, n \geq 1$, $x^m$ is an antipalindrome if and only if $x^n$ is an antipalindrome.
\end{proposition}

\begin{proposition} \label{pro:auv}
For all nonempty antipalindromes $u, v$, the word $uv$ is an antipalindrome iff both $u$ and $v$ are powers of an antipalindrome $z$.
\end{proposition}


\begin{proposition}
The word $x$ is an antipalindrome if and only if there exists a word $z$ such that $x = \overline{z} \, \sha \, z^R$.
\end{proposition}

A binary word $x$ is called an {\it antipalstar} if it is the
concatenation of $1$ or more antipalindromes. A nonempty antipalstar is
called {\it prime} if it is not the product of two or more even-length
antipalindromes. An antipalstar is always an antipalindrome, but the
converse is false, because for any antipalindromes $x, y$ the word
$xyx$ will be an antipalindrome. The following theorem is a counterpart
of the decomposition property for \emph{palstars} \cite{KMP77}.

\begin{theorem}
Every antipalstar can be factored uniquely as the concatenation of
prime antipalstars.
\end{theorem}

\begin{proof}
If some antipalstar has multiple factorizations into prime
antipalstars, then some prime antipalindrome $u$ has another prime
antipalindrome $v$ as a prefix. Let $v$ be the shortest antipalindrome
in such pairs. If $|v|\le |u|/2$, then $u=vzv$, where $z$ is either
empty or an antipalindrome; this contradicts the primality of $u$. Let
$|v|>|u|/2$, $u=x\overline x^R$, $v=y\overline y^R$. Then $v=xy$,
$\overline x^R=yz$ for some nonempty $y,z$. Since $v$ is an
antipalindrome, we have $v=\overline{xy}^R=\overline y^R\overline
x^R=\overline y^Ryz$. Since $\overline y^Ry$ is an antipalindrome, this
contradicts the minimality of $v$. Thus, no prime antipalindrome has
another prime antipalindrome as a prefix.
\qed \end{proof}

Next we look at the number of antipalindromic factors in a word.

\begin{theorem} \label{the:arich1}
A word $w$ of length $n \geq 1$ has at most $n-1$ distinct nonempty factors that are antipalindromes. 
\end{theorem}

\begin{proof}
For any nonempty antipalindromic factor of $w$, consider its leftmost
occurrence in $w$. We show that no two such occurrences end in the same
position in $w$. If they did, say $x$ and $y$, with $|x| < |y|$, then
$x$ is a suffix of $y$. But then $x^R=\overline{x}$ is a prefix of
$y^R=\overline{y}$. So $x$ is a prefix of $y$. Since $|x| < |y|$, we
have found an occurrence of $x$ to the left of its leftmost occurrence,
a contradiction.

Since no nonempty antipalindromic factor of $w$ ends at position 1, the
number of possible end positions of such factors is at most $n-1$,
whence the result.
\qed \end{proof}

One can introduce the notion of an \emph{a-rich word} as a word with
maximum possible number of antipalindromic factors. But the following
theorem shows that these words are trivial, in contrast with the rich
words.

\begin{theorem}
For all $n \geq 1$, there are exactly two a-rich words of length $n$. 
\end{theorem}

\begin{proof}
Let us prove that any word having the factor $00$ or $11$ is not
a-rich. Consider such a word $w$ and the position in which its leftmost
factor of the form $aa$ ends. An antipalindrome ending with $aa$ must
begin with $\overline {aa}$; hence, $w$ has no antipalindrome ending in
the chosen position. From the proof of Theorem~\ref{the:arich1} we know
that an a-rich word contains nonempty antipalindromes ending at every
position, except position 1, so $w$ is not a-rich.

Thus, only two words of each length $n$ remain. These words are
$(10)^k$ and $(01)^k$ if $n = 2k$ is even, and $(10)^k 1$ and $(01)^k
0$ if $n = 2k+1$ is odd.

To see that these words have $n{-}1$ distinct antipalindromes, consider a word $w=1010 \cdots$; the other word admits the same proof.  Note that $w[1..2k]=(10)^k$ and $w[2..2k{+}1]=(01)^k$ for $k \geq 1$ are the leftmost occurrences of antipalindromes. This gives an antipalindrome ending at every position except position $1$, for the  total of $n{-}1$. \qed
\end{proof}

We call a word $x$ {\it creaky} if it is a conjugate of its reversed complement $\overline{x}^R$.

\begin{proposition}
The word $x$ is in $\antipal^2$ if and only if it is creaky.
\end{proposition}

\begin{proof}
Suppose $x \in \antipal^2$.  Then $x = uv$ where $u, v$ are antipalindromes. So $\overline{x^R} = \overline{v^R} \overline{u^R} = vu$, a conjugate of $x$.

For the other direction, suppose $x$ is a conjugate of $\overline{x^R}$.  Then $x = uv$ and $\overline{x^R} = vu$. From $x = uv$ we get $x^R = v^R u^R$ and then $\overline{x^R} = \overline{v^R} \overline{u^R}$. It follows that $v = \overline{v^R}$ and $u = \overline{u^R}$. \qed
\end{proof}

\begin{proposition} \label{pro:xmn}
For all integers $m, n \geq 1$, the word $x^m$ is creaky iff  $x^n$ is creaky.
\end{proposition}

\begin{proof}
It suffices to prove the result for $m = 1$.  Suppose $x$ is creaky. Then there exist $u, v$ such that $x = uv$ and $\overline{x}^R = vu$. Then $x^n = (uv)^n$, which is clearly a conjugate of $(vu)^n = (\overline{x}^R)^n = \overline{(x^n)^R}$.

Suppose $x^n$ is creaky.  Then there exist $u, v$ such that $x^n = uv$ and $\overline{ (x^n)^R} = vu$.  Then there exist an integer $m$ and words $x', x''$ such that $u = x^m x'$, $v = x'' x^{n-m-1}$, where $x' x'' = x$.  So $vu = (x'' x')^n$, and it follows that $\overline{{x^n}^R} = x'' x'$,
a conjugate of $x$. \qed
\end{proof}

In analogy with palindromic pairs, we define factorization of a creaky word $w$ as a pair of antipalindromes $u$, $v$, where $v\ne\varepsilon$ and $uv=w$.
The following analog of Theorem~\ref{the:pp1} holds.

\begin{theorem} \label{the:app}
A creaky word $w$ has $m$ factorizations if and only if $w=z^m$ for a primitive word $z$.
\end{theorem}

\begin{proof}
A useful observation, used several times in this proof, is that if
$v\in\antipal$, then for any word $u$ any two of the following
conditions (i) $u\in\antipal$, (ii) $u$ is a prefix of $v$, (iii) $u$
is a suffix of $v$, imply the third one.

First we take a creaky word with two different factorizations,
$w=u_1v_1=u_2v_2$, and prove that it is not primitive. We can assume
$|u_1|<|u_2|$. If $u_1=\varepsilon$, then $w\in\antipal$. By
Proposition~\ref{pro:auv}, $w$ is a nontrivial power of an
antipalindrome. Now let $u_1\ne \varepsilon$. We prove the following
fact by induction on $|w|$: for some antipalindromes $x$ and $y$, one
has $w=(xy)^m$, $u_1=(xy)^jx$, $u_2=(xy)^kx$, where $0\le j<k<m$. The
base case is trivial (the shortest creaky words have a unique
factorization), so we proceed with the inductive step.

Since $u_1\in\antipal$, it is both a prefix and a suffix of $u_2$. By
Lemma~\ref{lysch} (ii) we have $u_1=(xy)^ix$, $u_2=(xy)^{i+1}x$ for
some words $x,y$ and some $i\ge 0$. Similarly we have
$v_2=(x'y')^{i'}x'$, $v_1=(x'y')^{i'+1}x'$. Clearly,
$x,y,x',y'\in\antipal$. Furthermore, the suffix $yx$ of $u_2$ coincides
with the prefix $x'y'$ of $v_1$. So we have $w=(xy)^{i+i'+1}xx'$. If
$x'=y$, then we are done, so assume that $|x'|<|y|$ (the case
$|x'|>|y|$ is similar). Since $yx=x'y'$ are different factorizations of
a creaky word which is shorter than $w$, we apply the inductive
hypothesis to get $yx=(zt)^m$, $x'=(zt)^jz$, $y=(zt)^kz$. Then we have
$u_1=(tz)^{mi+m-k-1}t$, $u_2=(tz)^{m(i+1)+m-k-1}t$,
$w=(tz)^{m(i+i'+1)+m-k+j}$. The inductive step is finished.

We proved that a primitive creaky word has a unique factorization. Now let us take $w=z^m$ for a primitive $z$ and let $z=z'z''$ be the factorization of $z$. Then clearly all words of the form $z^iz'=(z'z'')^iz'$ and $z''z^i=z''(z'z'')^i$ are antipalindromes. Thus, $w$ has $m$ factorizations of the form $z^iz'\cdot z''z^{m-i-1}$. Conversely, suppose that $w$ has a factorization $z^i\hat z\cdot \tilde zz^{m-i-1}$. Then $\hat z,\tilde z\in\antipal$, implying $\hat z=z'$ and $\tilde z=z''$. Thus, the number of factorizations is exactly $m$.\qed
\end{proof}

Using Theorem~\ref{the:app}, it is easy to relate the growth function of $\antipal^2$ (sequence A045655 in the OEIS \cite{Slo}) to palindromic pairs.

\begin{theorem}
$C_{\antipal^2} (n)= \E(n,2)$ for all $n \geq 0$.
\end{theorem}

\begin{proof}
Antipalindromes can be mapped to even-length palindromes with the
bijection negating the right half of a word. We can extend this idea to
creaky words and even palindromic pairs. For a creaky word $w=z^m$ with
$z$ primitive, we know that $z$ is creaky (Proposition~\ref{pro:xmn})
and has a unique factorization $z=uv$ (Theorem~\ref{the:app}). Negating
the right halves of both $u$ and $v$, we get an even palindromic pair
$\hat z=\hat u\hat v$. The even palindromic pair $\hat w=\hat z^m$
will be the image of $w$. Theorem~\ref{the:pp1} allows one to invert
this mapping and thus obtain a bijection between creaky words and even
palindromic pairs of a fixed length.
\qed \end{proof}

\bibliographystyle{splncs03}
\bibliography{my_bib}

\end{document}